\documentclass{ecai2014}
\usepackage{times}
\usepackage{graphicx}
\usepackage{latexsym}
\usepackage{graphicx,amsmath}
\usepackage{colortbl}
\usepackage{arydshln}
\usepackage{color}
\usepackage{dsfont}
\usepackage{helvet}
\usepackage{courier}
\usepackage{amssymb}

\newtheorem{lemma}[theorem]{Lemma}
\newtheorem{corollary}[theorem]{Corollary}

\newcommand{\blackslug}{\penalty 1000\hbox{
    \vrule height 8pt width .4pt\hskip -.4pt
    \vbox{\hrule width 8pt height .4pt\vskip -.4pt
          \vskip 8pt
      \vskip -.4pt\hrule width 8pt height .4pt}
    \hskip -3.9pt
    \vrule height 8pt width .4pt}}

\newenvironment{proof}{$\;$\newline \noindent {\sc Proof.}$\;\;\;$\rm}{\qed}
\newcommand{\qed}{\hspace*{\fill}\blackslug}

\def\boxit#1{\vbox{\hrule\hbox{\vrule\kern4pt
  \vbox{\kern1pt#1\kern1pt}
\kern2pt\vrule}\hrule}}
\newcommand{\nph}{{$\mathcal{NP}$-hard}}

\newcommand{\fpt}{{$\mathcal{FPT}$}}

\newcommand{\para}{{t}}
\newcommand{\ins}{{I}}

\begin{document}
\title{Election Attacks with Few Candidates}
\author{Yongjie Yang
\institute{Universit\"at des Saarlandes, Germany}}

\maketitle
\begin{abstract}
We investigate the parameterized complexity of strategic behaviors in generalized scoring rules. In particular, we prove that the manipulation, control (all the 22 standard types), and bribery problems are fixed-parameter tractable for most of the generalized scoring rules, with respect to the number of candidates. Our results imply that all these strategic voting problems are fixed-parameter tractable for most of the common voting rules, such as Plurality, $r$-Approval, Borda, Copeland, Maximin, Bucklin, etc., with respect to the number of candidates.
\end{abstract}

\section{Introduction}
Voting has been recognized as a common approach for preference aggregation and collective decision making whenever there exists more than one
alternative for a community to choose from. It comes with a  wide variety of applications which ranges from
multi-agent systems, political elections, recommendation systems, machine learning etc. \cite{DBLP:journals/cj/PittKSA06,DBLP:conf/hci/Popescu13b,DBLP:conf/atal/Xia13}. Unfortunately, due to the Gibbard-Satterthwaite theorem \cite{Gibbard73,Satterthwaite75}, any voting system which satisfies a set of desirable criteria is not strategy-proof, that is, there exists a voter who can make himself better off by misreporting his vote. To address this issue, Bartholdi et al. \cite{BARTHOLDI89} introduced the computational complexity to the study of strategic voting problems. The point is that if it is {\nph} to successfully perform a specific strategic behavior, the strategic individual(s) may give up performing such a strategic behavior.
Since then, exploring the complexity of strategic behaviors in voting systems has been one of the main focus of computational social choice community. We refer to \cite{DBLP:conf/birthday/BetzlerBCN12,DBLP:conf/sofsem/ChevaleyreELM07,Lindner08} for comprehensive surveys on this topic.

Recently, this purely worst-case analysis, which ignores real-world settings, was criticized by researchers. See \cite{DBLP:conf/aaai/ConitzerS06,DBLP:journals/aim/FaliszewskiP10,DBLP:conf/aldt/MatteiW13,DBLP:journals/jair/ProcacciaR07,DBLP:journals/amai/Walsh11} for detailed discussions. For their purpose, they proposed diverse measurements to evaluate the feasibility of strategic behaviors in practical elections. For example, Procaccia and Rosenschein \cite{DBLP:journals/jair/ProcacciaR07} introduced
the concept of junta distributions
(generally speaking, these are distributions over the elections that satisfy several constraints.)
and argued that
if an (heuristic) algorithm often solve the manipulation problem when the instances are
distributed according to a junta distribution, it would
also often solve the manipulation problem when the
instances are distributed according to many other plausible
distributions.

In this paper, we study the strategic behaviors in a variety of voting systems from the parameterized complexity perspective. The parameterized complexity was first systematically introduced by Downey and  Fellows \cite{fellows99}. Differently from the classical complexity, the parameterized complexity deals with problems in two dimensions. More specifically, an instance of a {\it{parameterized problem}} consists of a main part
and a parameter ${\para}$ which is normally a positive integer.
The main task in the parameterized complexity is to explore how the parameters affect the complexity of the problems. It turned out that under the framework of the parameterized complexity, many {\nph} problems become tractable with respect to specific parameterizations. More precisely, many {\nph} problems turned out to be solvable in $f(\para)\cdot |{\ins}|^{O(1)}$ time. Here, $f$ is a computable function that depends only on the parameter $\para$. All the parameterized problems which fall into this category are called fixed-parameter tractable ({\fpt} for short). However, the parameters do not always behave in this way. There are parameterized problems which do not admit {\fpt}-algorithms unless the parameterized complexity hierarchy collapses at some level, which is commonly believed to be unlikely. This discussion is beyond our focus in this paper. For a comprehensive understanding of parameterized complexity, we refer to the text of Niedermeier \cite{rolf06}.
For recent developments of parameterized complexity applied to computational social choice, we refer to \cite{DBLP:conf/birthday/BetzlerBCN12,Lindner08}.

A natural parameter in the voting scenario is the number of candidates. This parameter is relatively small in some real-world settings. For example, a political election normally contains only a few candidates. A reference library of preference data assembled by Mattei and Walsh \cite{DBLP:conf/aldt/MatteiW13} also reveals such a situation. Out of their 14 sets of election data from the real-life settings, 5 data sets contain less than 10 candidates each (\today).

In this paper, we aim at deriving a general framework for achieving {\fpt} results with respect to the number of candidates.
To this end, we adopt the concept of the class of generalized scoring rules which was introduced by Xia and Conitzer \cite{DBLP:conf/sigecom/XiaC08a}. In particular, we prove that the manipulation, control (all the 22 standard types) and bribery problems are {\fpt} for most of the generalized scoring rules, with respect to the number of candidates. Since many common voting rules fall into the category of the generalized scoring rules, these tractability results hold for these voting rules, among which are all the positional scoring rules (e.g., Borda, $r$-Approval, Veto, Plurality), Copeland$^{\alpha}$, Maximin, Bucklin, Ranked pairs, Schulze, Nanson's and Baldwin's. 
\bigskip

\noindent{\bf{Related Works.}} Hemaspaandra et al. \cite{DBLP:conf/atal/HemaspaandraLM13} recently studied the manipulation, control and bribery problems in Schulze's and Ranked pairs voting systems. They proved that all these strategic problems in Schulze and Ranked pairs voting systems are {\fpt} with respect to the number of candidates. Gaspers et al.~\cite{DBLP:conf/atal/GaspersKNW13} proved that the manipulation problem in Schulze voting system is indeed polynomial-time solvable for any number of manipulators.  Faliszewski et al. \cite{DBLP:journals/jair/FaliszewskiHHR09} studied Copeland$^{\alpha}$ control problems and achieved {\fpt} results for most of the control problems in Copeland$^{\alpha}$ voting with respect to the number of candidates, for every $0\leq \alpha\leq 1$. Besides the manipulation, (22 standard forms of) control and bribery problems, many other strategic voting problems were also studied from the parameterized complexity perspective by researchers. Faliszewski et al. \cite{DBLP:journals/jair/FaliszewskiHH11} studied a multimode control problem (in this model, the strategy individuals are allowed to add votes, delete votes, add candidates, delete candidates, and change votes simultaneously) and proved that this problem is {\fpt} with respect to the number of candidates for voting rules which are integer-linear-program implementable.
Dorn and Schlotter \cite{DBLP:journals/algorithmica/DornS12} proved that the swap bribery problem is {\fpt} with respect to the number of candidates for any voting system which is described by linear inequalities. 
Betzler et al. \cite{DBLP:conf/ijcai/BetzlerHN09} proved that the possible winner problem is {\fpt} with respect to the number of candidates for Maximin, Copeland and Ranked pairs voting rules.
Elkind et al. \cite{DBLP:conf/atal/ElkindFS10} devised a general framework for classifying the fixed-parameter tractability of the winner determination problems for voting rules which are ``distance-rationalizable''. For parameterized complexity of strategic voting problems with respect to other parameters than the number of candidates, we refer to \cite{DBLP:conf/birthday/BetzlerBCN12} for a survey.

\section{Preliminaries}\label{prelim}
{\bf{Common Rules}}. We follow the terminology of the work of Xia and Conitzer \cite{DBLP:conf/ijcai/XiaC09}. Let $\mathcal{C} = \{c_1, . . . , c_m\}$ be a set of candidates. A {\it{linear order}} on $\mathcal{C}$ is a transitive,
antisymmetric, and total relation on $\mathcal{C}$. The set of all linear
orders on $\mathcal{C}$ is denoted by $L(\mathcal{C})$. An $n$-{\it{voter profile}} $P$ on $\mathcal{C}$ consists of $n$ votes defined by linear orders on $\mathcal{C}$. That is, $P = (V_1, . . . , V_n)$, where for every $i\leq n, V_i\in L(\mathcal{C})$. Each vote represents the preferences of the respective voter over the candidates. In particular, a candidate $c$ is ranked higher than another candidate $c'$ in a vote, if the voter prefers $c$ to $c'$. For convenience, we also use $\succ_i$ to denote a vote.  Throughout this paper, we use the words vote and voter interchangeably. The set of all profiles
on $\mathcal{C}$ is denoted by $P(\mathcal{C})$. In the remainder of the paper, $m$
denotes the number of candidates and $n$ denotes the number
of voters. A {\it{(voting) rule}} is a function that maps a voting profile to a single candidate, the winner.

\begin{itemize}
\item Positional scoring rules. Every candidate gets a specific score from each vote according to the position of the candidate in the vote. More specifically, a scoring voting rule is defined by a scoring vector $\vec{\lambda}=(\lambda_1,\lambda_2,...,\lambda_m)$ with $\lambda_1\geq \lambda_2\geq,...,\geq \lambda_m$.
    The candidate ranked in the $i$-th position in a vote gets $\lambda_i$ points from this vote. The winner is the candidate with the highest score.\footnote{If more than one candidate has the highest score, we break the tie by a fixed deterministic tie-breaking rule. This applies to all the other voting rules discussed in this paper.} Following are some well-known positional scoring rules.

\begin{tabular}{l|l}
voting rules & scoring vectors \\ \hline

Borda & $(m-1,m-2,...,0)$ \\

$r$-Approval & $(1,...,1,0,...,0)$ with exactly $r$ many 1's. \\

Plurality & $(1,0,0,...,0)$ \\

Veto & $(1,1,...,1,0)$ \\
\end{tabular}


%


%
%


\item Maximin. For two candidates $c$ and $c'$, let $N(c,c')$ denote the number of votes which prefer $c$ to $c'$. We say $c$ beats $c'$ if $N(c,c')>N(c',c)$. The maximin score of a candidate $c$
is defined as $\min_{c'\in \mathcal{C}\setminus \{c\}}N(c,c')$. The winner is the candidate with the highest maximin score.

\item Copeland$^\alpha$.
Each candidate is compared with every other candidate. In each comparison, the one which beats its rival gets one point and its rival gets zero points. If they are tied, both get $\alpha$ points. The winner is the candidate with the highest score.

\item Instant-runoff (STV): If a candidate is ranked in the first position by more than half of the votes, the candidate wins.
Otherwise, the candidate which is ranked in the first position by the least number of votes is eliminated.
This is repeated until there is a candidate which is ranked in the first position by more than half of the votes.
\end{itemize}
\bigskip

{\bf{Generalized Scoring Rules.}} In the following, we give the definition of the class of the generalized scoring rules which was introduced by Xia and Conitzer \cite{DBLP:conf/ijcai/XiaC09}.

Let $K = \{1, . . . , k\}$. For any $\vec{a},\vec{b} \in \mathbb{R}^k$, we say that $\vec{a}$ and $\vec{b}$ are
{\it{equivalent}} with respect to $K$, denoted by $\vec{a} \sim _K \vec{b}$, if for
any $i, j \in K, \vec{a}[i] > \vec{a}[j] \Leftrightarrow \vec{b}[i] > \vec{b}[j]$ and $\vec{a}[i] < \vec{a}[j] \Leftrightarrow \vec{b}[i] < \vec{b}[j]$ (where $\vec{a}[i]$ denotes the $i$-th component of the vector $\vec{a}$, etc.).

A function $g : \mathbb{R}^k \rightarrow \mathcal{C}$ is {\it{compatible}} with $K$ if
for any $\vec{a},\vec{b} \in \mathbb{R}^k, \vec{a} \sim _K \vec{b} \Rightarrow g(\vec{a})
= g(\vec{b})$.

Let $k \in \mathbb{N}, f : L(\mathcal{C})\rightarrow \mathbb{R}^k$, and $g : \mathbb{R}^k \rightarrow \mathcal{C}$ where $g$ is compatible with
$K$. The functions $f$ and $g$ determine the
{\it{generalized scoring rule}} $GS(f, g)$ as follows. For any profile of votes
$V_1, . . . , V_n \in L(\mathcal{C}), GS(f, g)(V_1, . . . , V_n) =g(\sum_{i=1}^n f(V_i))$. That is, every vote results in a vector of scores according to
$f$, and $g$ decides the winner based on comparisons between the total scores. Here we call $f$ the {\it{generalized scoring function}} and $g$ the {\it{decision function}}. Moreover, we say that $GS(f, g)$ is of {\it{order}} $k$ and $\sum_{i=1}^n f(V_i)$ is the {\it{total score vector}} of the profile according to $GS(f,g)$. For convenience, we also use $f(P)$ to denote $\sum_{i=1}^n f(V_i)$, where $P$ is the profile with the votes $V_1,...,V_n$.

Many common voting rules fall into the category of the generalized scoring rules. For example, for the Borda voting rule, the corresponding generalized scoring rule is specified as follows.

$k_{Borda} = m$.

$f_{Borda}(V) = (s(V, c_1), . . . , s(V, c_m))$, where $s(V,c_i)$ is the score of $c_i$ from the vote $V$.

$g_{Borda}(f_{Borda}(P)) =argmax_i (f_{Borda}(P))$, that is, the winner is the one with highest Borda score.

We point out that the class of generalized scoring rules also encapsulates many runoff voting rules, that is, voting rules where the winners are determined via several rounds. A typical example is the STV voting rule. Moreover, the definition of generalized scoring rules can be
generalized to voting rules selecting more than one winner \cite{DBLP:conf/sigecom/XiaC08a}. The following lemma summarizes the common voting rules known to fall into the category of generalized scoring rules.

\begin{lemma}\label{commonrulesfall}
\cite{procaccia13jair,DBLP:conf/ijcai/XiaC09} The following voting rules are generalized scoring rules: all the positional scoring rules, Copeland$^{\alpha}$, Maximin, STV, Baldwin's, Nanson's, Ranked pairs, Bucklin.
\end{lemma}
%
%
%
%
%

A common voting rule which is precluded by the class of generalized scoring rules is the Young's voting rule \cite{DBLP:conf/ijcai/XiaC09}. Goldsmith~et~al.~\cite{AAAI14a} recently studied a new class of voting
rules (rank-dependent scoring rules, RDSRs for short) and showed by an example that there are voting rules in this class which do not fall into the category of of the generalized scoring rules.

\bigskip

{\bf{Strategic Behaviors.}} We make use of the standard definitions of strategic behaviors in computational social choice. In the following, we briefly introduce the problems discussed in this paper. We refer to \cite{DBLP:journals/mlq/ErdelyiNR09,DBLP:journals/jair/FaliszewskiHHR09} for all the detailed definitions, including the manipulation, bribery and all the 22 standard control problems. In all these problems, we have as input a set $\mathcal{C}\cup \{p\}$ of candidates where $p$ is a distinguished candidate, and a profile $P=\{V_1,...,V_n\}$ of votes. The question is whether the distinguished candidate $p$ can become a winner (in this case, $p$ is not the winner in advance) or become a loser (in this case, $p$ is the winner in advance) by imposing a specific strategic behavior on the voting. The former case of making $p$ a winner is called a {\it{constructive}} strategic behavior, and the latter case is called a {\it{destructive}} strategic behavior. Observe that if the problem of a specific constructive strategic behavior is {\fpt} with respect to the number of candidates, so is the corresponding destructive case. To check this, suppose that we have an {\fpt} algorithm $Algo$ for a specific constructive strategic behavior problem. Then, we can guess a candidate $p'\in \mathcal{C}$ and run the algorithm $Algo$ but with the distinguished candidate being $p'$. Since we have at most $m$ guesses, the destructive case is solved in {\fpt}-time. Due to this fact, we consider only the problems of constructive strategic behaviors.
\medskip

{\bf{Manipulation.}} 
In addition to the aforementioned input, we have a set $\mathcal{V}'$ of voters who did not cast their votes yet. We call these voters {\it{manipulators}}. The question is whether the manipulators can cast their votes in a way so that $p$ becomes a winner.
\medskip

{\bf{Bribery.}} 
The bribery problem asks whether we can change at most $\kappa$ votes (in any way but still linear orders over the candidates) so that $p$ becomes a winner, where $\kappa\in \mathbb{N}$ is also a part of the input.
\medskip

{\bf{Control.}} There are 11 standard constructive control behaviors in total. Among them 7 are imposed on the candidate set and 4 are imposed on the vote set. We first discuss the candidate control cases. In these scenarios, we either add some candidates (limited or unlimited), or delete some candidates, or partition the candidate set into two sets (runoff or non-runoff partitions with ties-promote or ties-eliminate models). Since the number of the candidates is bounded by the parameter $m$, we can enumerate all the possibilities of performing the control strategic behaviors in {\fpt}-time with respect to $m$. Thus, if the winner is computable in {\fpt} time (which holds for all the common voting rules studied in this paper) with respect to $m$, the candidate control problems are {\fpt}. In the following, we restrict our attention to the vote control problems.

%
%
%
\smallskip

{\bf{Deleting votes:}} The problem of control by deleting votes asks whether we can remove at most $\kappa$ votes from the given profile so that $p$ becomes the winner, where $\kappa\in \mathbb{N}$ is also a part of the input.
\smallskip

{\bf{Partition votes:}} In the control by partitioning of votes, we are asked the following question: is there a partition of $P$ into $P_1$ and $P_2$ such that $p$ is the winner of the two-stage election where the winners of election $(\mathcal{C}\cup \{p\},P_1)$ compete
against the winners of $(\mathcal{C}\cup \{p\},P_2)$? We distinguish the ties-promote model and the ties-eliminate model. In the ties-promote model, all the candidates which are tied as winners in the first-stage election 
are promoted to the second stage election. In the ties-eliminate model, if there is more than one winner, then all these winners will not be moved to the second stage election.

{\bf{Adding votes:}} In addition to the aforementioned input, we have another list $P'$ of {\it{unregistered votes}}, and are asked whether we can add at most $\kappa$ votes in $P'$ to $P$ so that the distinguished candidate $p$ becomes the winner.


\section{The General Framework}
In this section, we investigate the parameterized complexity of strategic behaviors under the class of generalized scoring rules. Our main result is summarized in the following theorem.

\begin{theorem}\label{thmframework}
Let $\varphi=(f,g)$ be a generalized scoring rule of order $k$, where $f$ is the generalized scoring function and $g$ is the decision function. If $k$ is bounded by a function of the number of candidates, and $f$ and $g$ are computable in \fpt-time with respect to the number of candidates, then the manipulation, bribery and all the 22 standard control problems are \fpt under $\varphi$, with respect to the number of candidates.
\end{theorem}
\begin{proof}
According to the discussion in Sec. \ref{prelim}, we can restrict our attention to the constructive strategic behaviors of manipulation, bribery, control by adding/delting/partition votes. We derive \fpt-algorithms for these problems. Our algorithms rely on the theorem by Lenstra \cite{lenstra83}, which implies that the integer linear programming (ILP) is {\fpt} with respect to the number of variables. Specifically, we reduce the instances of the stated problems to instances of ILP with the number of variables bounded by some function in $m$.

Let $\psi$ be the function in $m$ with $k\leq \psi({m})$. Due to the definition of the generalized scoring rules, we need focus on at most $3^{k\choose 2}$ different types of total score vectors (for each pair of subindices $i,j\in \{1,2,...,k\}$, we have either $\vec{a}[i]>\vec{a}[j]$ or $\vec{a}[i]=\vec{a}[j]$, or $\vec{a}[i]<\vec{a}[j]$). Here we say two vectors $\vec{a}, \vec{b}\in \mathbb{R}^k$ have the same type if they are equivalent with respect to $K=\{1,2,...,k\}$. Since the decision function $g$ is computable in \fpt-time with respect to $m$, we can enumerate all the types of total score vectors in the final election (that is, elections after performing strategic behaviors) which result in $p$ being the winner. Each enumerated total score vector $\vec{a}$ is specified by, for each pair of subindices $i,j$, either $\vec{a}[i]>\vec{a}[j]$ or $\vec{a}[i]=\vec{a}[j]$, or $\vec{a}[i]<\vec{a}[j]$. Then, we reduce the subinstances to ILP instances. To this end, we assign variables to different types of votes and derive restrictions to ensure that the currently enumerated total score vector coincides with the final election. If the given instance is a true-instance, then at least one of the total score vector leads to a correct answer. We fix $\vec{a}$ as the currently enumerated total scoring vector. In the following, we show how to reduce these instances to ILP instances.
\medskip

{\bf{Manipulation.}} Let $P=(V_1,V_2,...,V_n)$ be the profile of non-manipulators, and let $\vec{b}=f(P)$ be the total score vector of $P$. Clearly, $\vec{b}$ can be calculated in {\fpt} time since the generalized scoring function is computable in {\fpt} time.
To reduce the manipulation problem to ILP, we assign variables to all the $m!$ possible linear orders over the candidates, one for each. Let $x_{\succ}$ denote the variable assigned to the linear order $\succ$. These variables indicate how many manipulators cast their votes which are defined as $\succ$. Now we introduce the restrictions.

(1) Let $t$ be the number of manipulators, we have

\[\displaystyle{\sum_{\succ}x_{\succ}}=t\]

Here, $\succ$ runs through all the linear orders in $L(\mathcal{C}\cup \{p\})$.

(2) For convenience, for each linear order $\succ$, we use $f_{\succ}[i]$ instead of $f(\succ)[i]$ to denote the $i$-th entry of the score vector of $\succ$ by the generalized scoring function $f$. For each pair $i,j\in \{1,2,...,k\}$ with $\vec{a}[i]-\vec{a}[j]\rhd 0$, where $\rhd\in \{>,=,<\}$ we have

\[\vec{b}[i]+\displaystyle{\sum_{\succ}(f_{\succ}[i]\cdot x_{\succ})}-\vec{b}[j]-\displaystyle{\sum_{\succ}(f_{\succ}[j]\cdot x_{\succ})}\rhd 0\]


{\bf{Bribery.}} 
We divide the votes into $P_{\succ_1},P_{\succ_2},...,P_{\succ_t}$ with $P_{\succ_i}$ containing all the votes defined as the linear order ${\succ_i}$.
For every two distinguished linear orders $\succ$ and $\succ'$, we assign a variable denoted by $x_{\succ}^{\succ'}$, which specifies how many voters from $P_{\succ}$ are bribed to recast their votes as ${\succ'}$. Clearly, we have at most $m!^2$ variables. 
For each $\succ$, let $N(\succ)$ be the number of the votes which are defined as $\succ$ after changing the votes according to the variables assigned to the instance. More precisely, $N(\succ)$ is given by 

\[N(\succ)=|P_{\succ}|-\displaystyle\sum_{\succ'\neq \succ}x_{\succ}^{\succ'}+\sum_{\succ'\neq \succ}x_{\succ'}^{\succ}\]

Now we introduce the restrictions. First, we have the following restriction since we can bribe at most $\kappa$ votes in total.

\[\displaystyle\sum_{\succ\neq \succ'} x_{\succ}^{\succ'} \leq \kappa\]
Here, $\succ$ and $\succ'$ with $\succ\neq \succ'$ run through all the linear orders over the candidates.

In addition, for each $P_{\succ}$, at most $|P_{\succ}|$ can be bribed. Hence, for each $P_{\succ}$, we have

\[\displaystyle\sum_{\succ'\neq \succ} x_{\succ}^{\succ'} \leq |P_{\succ}|\]

Finally, 
for every~$i,j$~\text{with}~$\vec{a}[i]- \vec{a}[j]\rhd 0$, where $\rhd \in \{>,<,=\}$, we have

\[\displaystyle\sum_{\succ} N(\succ)\cdot f_{\succ}[i]-\displaystyle\sum_{\succ} N(\succ)\cdot f_{\succ}[j]\triangleright 0\]
%
%

{\bf{Control by Adding/Deleting Votes.}} 
We first consider the adding votes case. We divide the unregistered votes into parts each containing all the votes defined as the same linear order. 
For each part containing the votes defined as linear order $\succ$, we assign a variable $x_{\succ}$, which specifies how many votes from this part are included in the solution. Now we introduce the restrictions. 

For each linear order $\succ$, let $N_1(\succ)$ be the number of registered votes defined as $\succ$, and $N_2(\succ)$ be the number of unregistered votes defined as $\succ$.
Since we can add at most $\kappa$ votes, we have the following restriction.

\[
\displaystyle\sum_{\succ}x_{\succ} \leq \kappa
\]

In addition, for each $\succ$, we have

\[x_{\succ}\leq N_2(\succ)\]

Finally, for every pair $i,j$ with $\vec{a}[i] -\vec{a}[j]\rhd 0$, where $\rhd \in \{>,<,=\}$, we have,

\[\displaystyle\sum_{\succ} (x_{\succ}+N_1(\succ))\cdot f_{\succ}[i]-\displaystyle\sum_{\succ} (x_{\succ}+N_1(\succ))\cdot f_{\succ}[j]\triangleright 0\]

The algorithm for the deleting votes case is analogous.

{\bf{Partition of Votes with Ties-Eliminate.}}
This case is slightly different from the above cases. First observe that $p$ has chance to be the final winner if $p$ is a temporary winner in at least one of the two subelections. 
Therefor, to solve the problem, we enumerate all possible candidates $p'$ which will compete with $p$ in the second-stage election. We immediately discard the enumerations for which $p$ is not the winner when competing with $p'$. The above procedure clearly takes polynomial time and leads to polynomially many subinstances, each asking whether we can partition the profile into two parts $P_1$ and $P_2$ so that $p$ is the winner in the voting with profile $P_1$, and $p'$ is the winner in the voting with profile $P_2$. Therefore, instead of enumerating all the possible total score vectors as discussed for the above controls, we enumerate all the possible vector pairs $\vec{a},\vec{b}\in \mathbb{R}^k$, where $\vec{a}$ is the potential total score vector for the voting profile $P_1$, and $\vec{b}$ is the potential total score vector for voting profile $P_2$. We discard all the enumerations for which $p$ (resp. $p'$) is not the winner with respect to $\vec{a}$ (resp. $\vec{b}$). Now, we adopt the similar method as discussed above to reduce each subinstance to an ILP instance. To this end, again we partition the votes into parts each containing all the votes defined as the same linear order over the candidates. We still use $x_{\succ}$ to denote the variable assigned to the part $P_{\succ}$ of votes defined as $\succ$. Here, $x_{\succ}$ indicates how many votes in $P_{\succ}$ go to $P_1$. 
The restrictions are as follows.
For each $\succ$, we have

\[x_{\succ}\leq |P_{\succ}|\]

For every pair $i,j$ with $\vec{a}[i]- \vec{a}[j]\rhd 0$, we have,

\[\displaystyle\sum_{\succ}(x_{\succ}\cdot f_{\succ}[i])-\displaystyle\sum_{\succ}(x_{\succ}\cdot f_{\succ}[j])\rhd 0\]

This equality is to ensure the that $p$ is the winning candidate in the voting with profile  $P_1$. The following equality is to ensure that $p'$ is the winning candidate in the voting with profile  $P_2$.

For every pair $i,j$ with $\vec{b}[i]- \vec{b}[j]\rhd 0$, where $\rhd \in \{>,<,=\}$, we have,

\[\displaystyle\sum_{\succ}(|P_{\succ}|-x_{\succ})\cdot f_{\succ}[i])-\displaystyle\sum_{\succ}(|P_{\succ}|-x_{\succ})\cdot f_{\succ}[j])\rhd 0\]

%
%
%
%
%

The proof for the ties-promote model is similar. However, in this case, we should adopt the multiwinner variant of the generalized scoring rules (which is possible  \cite{DBLP:conf/sigecom/XiaC08a}) as a tool. Besides, instead of enumerating a candidate $p'$, we need to enumerate all the pairwise disjoint subsets $C_1$ and $C_2$ of the candidates, with $p\in C_1$. If $p$ is not the winner in the election restricted to $C_1\cup C_2$, we discard the enumeration at the moment. Moreover, the restrictions are derived to making all the candidates in $C_1$ the co-winners in the voting profile $P_1$, and $C_2$ the co-winners in the voting profile $P_2$. \end{proof}

To use the framework described in Theorem \ref{thmframework}, we require that the order of the generalized scoring rule must be bounded by a function of the number of the candidates, and the scoring function and the decision function must be computable in {\fpt} time with respect to the number of candidates. In the following, we show that both the conditions are fulfilled for all the voting rules stated in Lemma \ref{commonrulesfall}.

\begin{lemma}\label{commonrulesfpt1}
For all the positional scoring rules, Copeland$^{\alpha}$, Maximin, STV, Baldwin's, Nanson's, Ranked pairs, Bucklin, the orders of the corresponding generalized scoring rules are bounded by functions of the number of candidates, and the scoring and decision functions of the corresponding generalized scoring rule are computable in {\fpt}-time.
\end{lemma}

\begin{proof}
We refer to \cite{DBLP:conf/sigecom/XiaC08a} for all the positional scoring rules, STV, Maximin, Ranked pairs and Copeland$^{\alpha}$ (\cite{DBLP:conf/sigecom/XiaC08a} described only for $\alpha=0.5$. However, with slight modification, the arguments work for all the $0\leq \alpha\leq 1$). In the following, we prove for the Baldwin's, Nanson's and Bucklin, by describing in detail the specifications of the respective generalized scoring rules. Hereby, let $\mathcal{C}$ denote the set of candidates.
\medskip


\medskip

{\bf{Bucklin.}} The Bucklin score of a candidate $c$ is the smallest
number $x$ such that more than half of the votes rank $c$ among
the top $x$ candidates. The winner is the candidate that has
the smallest Bucklin score. The respective generalized scoring rule is as follows.

$k_{Bucklin}=m^2$; the components are indexed by pairs $(c,i)$ where $c$ is a candidate and $i$ is a positive integer with $1\leq i\leq m$.

The score vector of a vote $\succ$ is calculated in the following way.

\[f_{Bucklin}(\succ)_{c,i}=
\begin{cases}
 1 & ~\text{if}~ c~\text{is ranked among the top}~i~\text{positions in}~\succ\\

 0 & ~otherwise~
\end{cases}\]

The decision function $g_{Bucklin}$ works as follows. First, sum up all the score vectors of the votes. Let $\vec{a}$ be the total score vector. It is clear that the Bucklin score of a candidate $c$ is the minimum value $i$ for which $\vec{a}[c,i]> n/2$, where $n$ is the number of votes. Then, the winner is the one with the minimum Bucklin score. 

%
%
%
\medskip

{\bf{Nanson's and Baldwin's.}} These two voting rules are multiround runoff rules, meaning that the winner is selected via rounds in each some candidates are removed from the election. Specifically, in the Nanson's voting, all the candidates with Borda score no greater than the average Borda score are
eliminated in each round. In the next round, the Borda scores of the remaining candidates are recomputed, as
if the eliminated candidates were not in the voting. This is repeated until there is a
final candidate left. The Baldwin's is similar to the Nanson's with difference that in each round the eliminated candidate is the one with least Borda score.

Xia and Conitzer \cite{DBLP:conf/sigecom/XiaC08a} proved that for any voting rule with finitely many runoff rounds, if in each step the rule used to rule out the eliminated candidates is also a generalized scoring rule, then the multiround runoff rule is a generalized scoring rule. Moreover, their (constructive) proof implies that if in each step the generalized scoring rule has {\fpt}-time computable functions $f'$ and $g'$, with respect to the number of candidates, the respective generalized scoring rule of the multiruound voting rule (with polynomially many rounds) also has a {\fpt}-time computable functions $f$ and $g$. We refer to Appendix 1 in \cite{DBLP:conf/sigecom/XiaC08a} for checking further details. Due to this fact, it is sufficient to show that in each round the procedure of selecting the eliminated candidates is a generalized scoring rule with bounded order,  {\fpt}-time computable decision function $g$ and {\fpt}-time computable scoring function $f$.

We consider first for the Nanson's voting. Let $(c_1,c_2,...,c_m)$ be a arbitrary fixed order of the candidates.

$k = m$.

$f(V) = (s(V, c_1), . . . , s(V, c_m))$, where $s(V,c_i)$ is the Borda score of $c_i$ from the vote $V$.

The decision function $g$ selects the winner(s) as follows. Let $avg=\frac{1}{m}\cdot \sum_{i=1}^m f(P)[i]$. The winners (the candidates which are eliminated) are the candidates $c_i$ with $f(P)[i]< avg$.

The Baldwin's voting is similar with the difference that the candidates $c_i$ with minimum ${f(P)[i]}$ are eliminated.
\end{proof}

Based on Theorem \ref{thmframework} and Lemma \ref{commonrulesfpt1}, we have the following corollary.

\begin{corollary}
All the manipulation, bribery and the 22 standard control problems for the following voting rules are {\fpt} with respect to the number of candidates: all the positional scoring rules, Copeland$^{\alpha}$, Maximin, STV, Baldwin's, Nanson's, Ranked pairs and Bucklin.
\end{corollary}

\section{Discussion}
The class of generalized scoring rules was first introduced by Xia and Conitzer \cite{DBLP:conf/sigecom/XiaC08a} to investigate the frequency of coalitional manipulability. In this scenario, the focus is how the probability of a random profile being manipulable changes as the number of manipulators increases from 1 to infinite. The class of generalized scoring rules was also used in investigating the margin of victory in voting systems~\cite{DBLP:conf/sigecom/Xia12}. Moreover, Xia and Conitzer \cite{DBLP:conf/ijcai/XiaC09} characterized the class of generalized scoring rules as the class of voting rules that are anonymous and finitely locally consistent. A highly related class of voting rules is the class of hyperplane rules introduced by Mossel et al. \cite{procaccia13jair}. Mathematically,
the generalized scoring rules are equivalent to the hyperplane rules \cite{procaccia13jair}.

In this paper, we extend the application of generalized scoring rules by exploring the parameterized complexity of strategic voting problems. In particular, we show that from the viewpoint of parameterized complexity, the manipulation, bribery and control problems which are {\nph} in many voting systems turned out to be fixed-parameter tractable ({\fpt}), with respect to the number of candidates. The key point of our {\fpt} algorithms is the compatibility of the decision function $g$ in the generalized scoring rules, which enables us to enumerate all the desirable total score vectors in {\fpt} time.

To date, many strategic problems have been proved {\nph}. A challenging task would be to explore the connections between these {\nph}ness via the notion of the generalized scoring rules. For this purpose, a deeper exploitation of the functions $f$ and $g$ is needed. Besides, there are also many other
problems in computational social choice that pertain to conducting strategic
voting, some of which were introduced quite recently \cite{DBLP:conf/atal/FaliszewskiHH13,DBLP:conf/atal/PerekFPR13}. Exploring the parameterized complexity of these newly proposed voting problems via the framework of generalized scoring rules is also an interesting topic.

{\bf{Note.}} Similar results of this paper were independently announced by Xia~\cite{lirongxiaaamas2014}. However, there are several differences. First, our results apply to all the 22 standard control problems, while the results in~\cite{lirongxiaaamas2014} does not include the control by partition votes. Second, Xia studied the winner determination problem which is not discussed in this paper.

\bibliography{D:/Dropbox/sociachoiceref}
\bibliographystyle{ECAI2014}
\end{document}